\documentclass[journal]{IEEEtran}
\usepackage{geometry} 
\usepackage{amsfonts} 
\usepackage{paralist}
\usepackage{subcaption}  
\usepackage{amssymb}
\usepackage{amsfonts}
\usepackage{array,booktabs}
\usepackage{commath}
\usepackage{amsmath}
\usepackage{amsthm}
\usepackage{comment}
\usepackage{pdfcomment}
\usepackage{enumitem}
\usepackage{tikz}
\usepackage{hyperref}
\usepackage{cite}
\usepackage{cleveref}
\usepackage[siunitx]{circuitikz} 
\usetikzlibrary{decorations,backgrounds}
\usepackage{ragged2e}    
\usepackage{multicol}

\DeclareMathOperator*{\argmin}{arg\,min}

\theoremstyle{definition}
\newtheorem{theorem}{Theorem}
\theoremstyle{definition}
\newtheorem{definition}{Definition}
\theoremstyle{definition}

\theoremstyle{definition}

\theoremstyle{definition}

\theoremstyle{definition}

\usepackage{pifont}
\newcommand{\cmark}{\ding{51}}%
\newcommand{\xmark}{\ding{55}}%

\usepackage[font={small}]{caption}
\graphicspath{{figures/}}
\usepackage{tabularx}

\begin{document}
\title{Decentralized Safety-Critical Control of Resilient DC Microgrids 
with Large-Signal Stability Guarantees}

\author{$^{1}$Muratkhan~Abdirash,~\IEEEmembership{Student Member,~IEEE,} and~$^{1}$Xiaofan~Cui,~\IEEEmembership{Member,~IEEE}
\thanks{$^{1}$M. Abdirash and X. Cui are with the Department
of Electrical and Computer Engineering, University of California, Los Angeles, 90095 USA email: {\tt\small mabdirash@ucla.edu, cuixf@seas.ucla.edu} (corresponding author). This work is based on work in part by the UCLA faculty start-up fund. This work is presented in part at the 2025 American Control Conference, Denver, CO, USA, July 8-10, 2025.}}
\maketitle
\begin{abstract}
The increasing penetration of distributed energy resources and power-electronics interfaces in DC microgrids, coupled with rising cyber threats, necessitates primary controllers that are provably safe, cyber-resilient, and practical. Conventional droop-based methods remain prevalent due to their simplicity, yet their design is largely empirical and conservative, lacking rigorous guarantees. Advanced strategies improve certain aspects, but often sacrifice scalability, robustness, or formal safety. In this work, we propose a Distributed Safety-Critical Controller (DSCC) that systematically integrates global stabilization with formal safety guarantees in a fully decentralized manner. Leveraging control barrier functions and the port-Hamiltonian system theory, the DSCC achieves scalable safe stabilization while preserving real-time implementability. High-fidelity switched-circuit simulations validate the controller’s advantages under various contingencies. This framework paves the way for resilient, safety-critical, and scalable control in next-generation DC microgrids.
\end{abstract}
\begin{figure}[ht]
\centering
\includegraphics[width=0.9\linewidth]{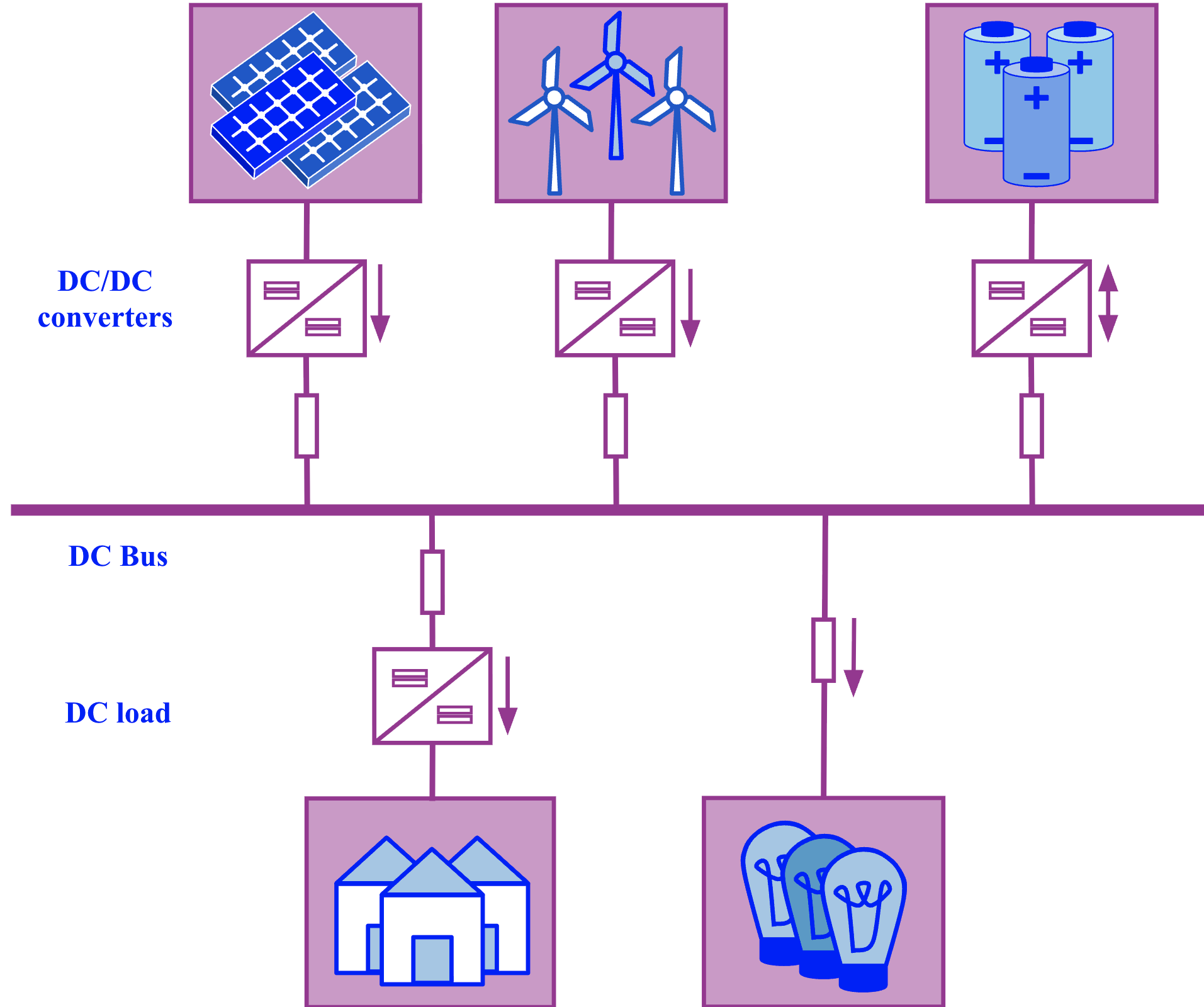} \caption{Illustrative example of a single-bus DC microgrid. It delivers the energy generated by DC power sources such as solar panels and batteries to DC loads such as electronic devices through one bus.}\label{fig:high-level}
\end{figure} 

\section*{Nomenclature}
\setlength{\tabcolsep}{4pt}
\begin{table}[ht]\centering
\caption*{}\small\begin{tabular}{cc}
\hline
$v_j$ &Output voltage of $j$-th converter [V]\\

$i_{s_j}$ &Current supplied to $j$-th DER-interfacing converter [A]\\

$i_{t_j}$ &Current through $j$-th transmission line [A]\\

$v_b$ &Bus voltage of the microgrid [V]\\

$i_f$ &Current through the filter inductor at the load [A]\\

$v_l$ &Voltage supplied to DC load [V]\\
\hline
$C_j$ &Output capacitance of $j$-th converter [mF]\\

$R_j$ &Resistance of $j$-th transmission line [m$\Omega$]\\

$L_j$ &Inductance of $j$-th transmission line [mH]\\

$L_f$ &Filter inductance at the load-interfacing converter [mH]\\

$C_b$ & Bus capacitance [mF]\\

$C_l$ & Output capacitance of the load-interfacing converter [mF]\\

$R_l$ & Resistive load to the microrgrid [$\Omega$]\\

$r_l$ & DC load resistance [$\Omega$]\\
\hline
DER &Distributed energy resources\\

KCL &Kirchhoff's current law\\

KVL &Kirchhoff's voltage law\\

CPL &Constant power load\\

CLF &Control Lyapunov Function\\

CBF &Control Barrier Function\\

QP &Quadratic program\\
\hline
\end{tabular}\label{tab:nomenclatuer}\end{table}
\section{Introduction}
DC microgrids have emerged as a promising paradigm for next-generation energy systems due to their inherent compatibility with distributed energy resources (DER), efficient interface with modern electronic loads, and the potential for improved overall energy efficiency \cite{6263309,6774539, 5587899}. A popular example of a DC microgrid is given in Fig. \ref{fig:high-level}, which describes the power grid of an industrial park with DER. Despite these advantages, physical resilience in DC microgrids remains challenging. Key difficulties include ensuring stability and safety in the presence of nonlinear loads, low system inertia, and millisecond-scale rapid transients introduced by DC/DC power electronic converters \cite{shahgholian2021}.

With its increasing deployment in residential and commercial energy infrastructures, cyber-resilience has also become an escalating concern. In July 2024, the FBI warned that attackers could remotely compromise residential converters to reduce solar output or manipulate home batteries, and similarly target solar farm microgrids to cause widespread outages \cite{ribeiro2024}.
\begin{table*}
\centering\small
    \begin{tabular}
{p{0.11\linewidth}p{0.09\linewidth}p{0.06\linewidth}p{0.06\linewidth}p{0.06\linewidth}p{0.065\linewidth}p{0.06\linewidth}p{0.06\linewidth}p{0.06\linewidth}p{0.05\linewidth}}
     \hline Control method & Control mode & Small-signal stability & Large-signal stability & Safety \& stability trade-off & Local optimality of control input & Comp-lexity & Model dependency & Adapt-ability & Refer-ences \\ \hline
        \textbf{Droop}  & centralized & \centering\cmark & \centering\cmark & \centering\xmark & \centering\xmark 
       & low & high & medium & \cite{herrera2015stability}\\ 
       \textcolor{white}{Droop}  & decentralized& \centering\cmark & \centering\xmark & \centering\xmark & \centering\xmark 
       & low & medium & low &\cite{AlrajhiAlsiraji2021}\\\hline 
       \textbf{MPC}  & centralized & \centering\cmark & \centering\cmark & \centering\cmark & \centering\cmark 
       & high & high & medium &\cite{xu2019offset}\\ 
       \textcolor{white}{MPC}  & decentralized & \centering\cmark & \centering\cmark & \centering\cmark & \centering\cmark 
       & high & high & medium &\cite{zhang2022decentralized}\\ \hline
       \textbf{ADRC}  & centralized & \centering\cmark & \centering\cmark & \centering\xmark & \centering\xmark 
       & medium & high & high& \cite{8668820}\\ 
       \textcolor{white}{ADRC}  & decentralized  & \centering\cmark & \centering\cmark & \centering\xmark & \centering\xmark 
       & medium & medium & medium& \cite{li2024decentralized}\\\hline 
       \textbf{Reinforcement learning} & centralized& \centering\cmark & \centering\cmark & \centering\cmark & \centering\cmark 
       & high & medium & medium & \cite{9817114}\\
       \textcolor{white}{n} & decentralized  & \centering\xmark & \centering\xmark & \centering\cmark & \centering\cmark 
       & high & low & medium & \cite{10018476}\\ \hline
       \textbf{Safety-critical}& centralized & \centering\cmark & \centering\cmark & \centering\cmark & \centering\cmark 
       & high & high & medium & \cite{abdirash2025nonlinearoptimalcontroldc}\\ \hline
    \end{tabular}
    \caption{State of the Art Stabilizing Controllers for a DC Microgrid}
    \label{controllers}
\end{table*}
Cyber attacks on DC microgrids can manifest across multiple layers of the system, including the communication, control, and hardware layers \cite{ahn2024}. At the communication layer, denial-of-service (DoS) attacks on the microgrid control center can disrupt coordination and lead to extended blackouts \cite{chen2022a}. Malware, such as worms, can also propagate rapidly across communication links, transforming isolated point failures into system-wide outages \cite{zhong2015}. These vulnerabilities have motivated research on fully decentralized operation schemes.

Significant efforts have been devoted to the cyber-resilience of decentralized secondary control in microgrids \cite{rajabinezhad2025, zuo2022}. However, comparatively less attention has been paid to the cyber-resilience of primary control, which ensures critical functions such as system stability and safety. Primary control must operate within a short time scale (on the order of milliseconds), making the design of computationally efficient and provably safe algorithms particularly challenging. 

In the control layer, power converters, as the backbone of the DC microgrid, are especially vulnerable to false data injection (FDI) attacks, such as setpoint spoofing \cite{zhang2023} and sensor data tampering \cite{roiggreidanus2024}. Beyond continuous false signal injection, recent studies reveal that attackers can employ stealthier strategies by intermittently injecting impulsive large-signal disturbances \cite{liu2021a}. These attacks are able to evade summation-based anomaly detectors \cite{liu2023} while still triggering nonlinear instability modes; therefore, mitigation of this task relies on the intrinsic large-signal stabilization capability of the primary controllers.

Protection devices are another critical vulnerability, where malware can alter relay settings and cause trigger failures \cite{liu2017b}. Safety-aware primary controllers can serve as an additional protection layer, cross-checking traditional devices to enhance reliability and enable anomaly detection. 

Motivated by these challenges, there is an urgent need to design new resilient control architectures for power converters in DC microgrids. In particular, new control schemes should provide theoretically provable guarantees on the safety and stability of the microgrid without compromising practicality. The majority of existing control algorithms with their safety and stability guarantees are summarized in Table \ref{controllers}. Droop controllers remain the most widely used practical control scheme for large-scale DC microgrids. However, as linear controllers, their design is largely empirical and oriented toward worst-case scenarios, often resulting in overly conservative performance \cite{marimuthu2023review}. Nonlinear droop variants can improve transient performance, but still lack the rigor of theoretical guarantees. Recent advances in robust droop control provide large-signal stability guarantees \cite{herrera2015stability}; however, these guarantees are typically associated with a limited region of attraction (ROA). As a result, the practical controller design becomes complicated and may fail under large disturbances.
Active Disturbance Rejection Control (ADRC) is another notable example of a robust nonlinear controller, yet it lacks a straightforward mechanism to guarantee safety. Model Predictive Control (MPC) offers large-signal stabilization with direct enforcement of safety constraints, but the infeasibility problem of nonlinear MPC limits its practicality. Reinforcement learning–based control has recently emerged as an end-to-end path, although its high complexity presents significant challenges for implementation in primary control. In contrast, existing safety-critical controllers can enforce safety in computationally effective ways, but most are centralized, which limits scalability in practical microgrid applications.

 It is evident that DC microgrids have yet to have a control framework that elegantly balances safety and stability specifications. Without a comprehensive and intentional design of a controller, DC microgrids are more vulnerable to cyberphysical disturbances. In line with that, this work proposes a control approach that offers the following novel contributions
\begin{enumerate}\item Large signal stability enforced by a global control Lyapunov function obtained using Port-Hamiltonian system theory; \item Safety adherence enforced by control barrier functions, novel to power electronics community;
\item Fully decentralized optimal control scheme that combines contributions 1 and 2 and updates the controller in real time using a standard convex optimization solver.\end{enumerate}

Following this introduction, the technical discussion in Section II delves into the control-theoretic definitions of safety and stability for nonlinear systems. Section III models the dynamics of the DC microgrid. Section IV defines the optimal desired steady-state behavior. Section V solves for a nominal stabilizing controller that is robust to large signal disturbances such as DoS cyber attacks. Section VI proposes \textit{Decentralized Safety-Critical Controller} (DSCC), capable of avoiding hazardous safety thresholds of the circuitry while preserving large-signal stability. Section VII validates the DSCC on a high-fidelity power electronics simulation platform and compares its performance with the nominal controller. Finally, Section VIII concludes the paper.

\section{Important Notions and Tools from Control Theory}
\begin{definition}\label{def:control_aff} Consider a \textit{control-affine} system \begin{equation}
\dot{x}=f(x)+g(x)u,\label{eq:aff}
\end{equation}
where its state and control input are given by $x\in X\subseteq\mathbb{R}^n$ and $u\in U\subseteq\mathbb{R}^m$. A pair $(u^*,x^*)\in U\times X$ is an \textit{equilibrium} of system \eqref{eq:aff} if\begin{equation}f(x^*)+g(x^*)u^*=0_{n}.\label{eq:equil}
    \end{equation}
\end{definition}
\begin{definition}\label{def:lie_deri} Let the gradient operator $\partial(\cdot)$ be with respect to $x$ and in column vector form. A \textit{Lie derivative} is an operator acting on a function $V$ and is defined as $\mathcal{L}_fV:=\partial Vf(x)$. 
\end{definition}
\begin{definition}\label{def:lipschitz}Let $V:X\rightarrow\mathbb{R}$ be a real-valued function. $V$ is \textit{Lipschitz continuous} on the domain $X$, if there exists a constant $M>0$ such that \begin{equation}
        ||V(x)-V(y)||\leq M||x-y||,\forall x,y\in X,\label{eq:lipschitz}
    \end{equation}
    where $||\cdot||$ is the Euclidean norm.
\end{definition}
\begin{definition}\label{def:glo_sta} Let $k(x)$ be a Lipschitz continuous feedback control law for \eqref{eq:aff} with $k(x^*)=u^*$. The equilibrium $x^*$ of the closed-loop system \begin{equation}\dot{x}=f(x)+g(x)k(x),\label{eq:closedloop}\end{equation} is \textit{(locally) exponentially stable}, if there exist constants $M,\lambda>0$ such that \begin{equation*}
        ||x(t)-x^*||\leq M||x_0-x^*||e^{-\lambda t},
    \end{equation*} $\forall t>0$ and $\forall x_0\in(D\subset X)\text{ } X$ where $x(t)$ is the solution of \eqref{eq:closedloop} to an initial condition $x_0$. 
\end{definition}
\begin{definition}\label{def:CLF} Let $V:X\rightarrow\mathbb{R}_{\geq0}$ be continuously differentiable, $V(x)$ is a \textit{Control Lyapunov Function (CLF)} for \eqref{eq:aff}, if there exist $\alpha_1,\alpha_2>0$ and $\alpha_3>0$ such that \begin{subequations}\label{eq:clf}\begin{align}\alpha_1&||x-x^*||^2\leq V(x)\leq \alpha_2||x-x^*||^2,\forall x\in X;\\\mathcal{L}_gV(x)&=0\Rightarrow \mathcal{L}_fV(x)+\alpha_3||x-x^*||^2<0,\forall x\in X\setminus x^*.\end{align}\end{subequations}\end{definition}
\begin{definition}\label{def:ccp} A CLF $V$ for system \eqref{eq:aff} satisfies the \textit{continuous control property} (CCP) if for any $\epsilon>0$, there exists an open set $E\subseteq X$ such that for any $x\neq x^*\in E$, there exists a control input $u\in\mathbb{R}^m$ satisfying $||u-u^*||<\epsilon$ and \begin{equation}
        \dot{V}(x,u)<-\alpha_3||x-x^*||^2.\label{eq:ccp}
    \end{equation}
\end{definition}
\begin{definition}\label{def:safe set} Let $\mathcal{C}\subset X$ denote the safe region of \eqref{eq:aff}. The closed-loop system \eqref{eq:closedloop} is safe under $k(x)$, if $\mathcal{C}$ is rendered \textit{forward invariant} , i.e. if $x_0\in\mathcal{C}$ then $ x(t)\in\mathcal{C}, \forall t>0.$ 
\end{definition}
\begin{definition}\label{def:CBF}
    Let $\mathcal{C}$ be defined as the 0-superlevel set of continuously differentiable $h:X\rightarrow\mathbb{R}$ \begin{subequations}\label{eq:safe_set}\begin{align}
        \mathcal{C}&:=\{x\in X|h(x)\geq0\},\\
        \text{int}(\mathcal{C})&:=\{x\in X|h(x)>0\},\\
        \partial(\mathcal{C})&:=\{x\in X|h(x)=0\},\end{align}\end{subequations} where $\text{int}(\mathcal{C})$ is the interior of safe set $\mathcal{C}$ and $\partial(\mathcal{C})$ is the boundary. The reciprocal function $B:=1/h(x)$ is a \textit{Control Barrier Function (CBF)} for \eqref{eq:aff} if there exists $\beta>0$ such that
    \begin{equation}\mathcal{L}_gB(x)=0\Rightarrow \mathcal{L}_fB(x)-\frac{\beta}{B(x)}<0,\forall x\in\text{int}(\mathcal{C}).\label{eq:cbf_con}\end{equation}\end{definition}
\begin{definition}\label{def:PH} A \textit{Port-Hamiltonian} (PH) dynamical system is given by\begin{subequations}\label{eq:PH}
   \begin{align}
       \dot{x}&=(\mathcal{J}(x)-\mathcal{R}(x))\partial H(x)+g(x)u+g_z(x)z,\\
       y&=g_z(x)^T\partial H(x),\end{align}\end{subequations} where skew-symmetric $\mathcal{J}(x)^T=-\mathcal{J}(x)$ is called an interconnection matrix, and positive definite $\mathcal{R}(x)^T=\mathcal{R}(x)\geq0$ is called a dissipation matrix. The function $H$ is called a \textit{Hamiltonian} and is defined as $H(x)=x^TQx$ with positive-definite $Q^T=Q>0$. Let $z\in Z\subseteq\mathbb{R}^p$ denote the input that enters the system modulated by matrix $g_z(x)$, and $y\in Y\subseteq\mathbb{R}^p$ denote the output.\end{definition}\begin{definition}\label{def:global_pH} Consider $k$ interconnected PH subsystems (Def. \ref{def:PH} indexed by $j=1,\cdots,k$) such that all of their input and output ports become internal, i.e. $\sum_{j=1}^ky_j^Tz_j=0$. Then the \textit{Global Interconnected Port-Hamiltonian System (GIPH)} is given by \begin{subequations}\label{eq:global}\begin{align}
            x&:=\text{col}(x_1\cdots x_k),\\
            u&:=\text{col}(u_1\cdots u_k),\\
            \mathcal{R}(x)&:=\text{diag}\left (\mathcal{R}_1(x_1)\cdots\mathcal{R}_k(x_k)\right ),\\
            \mathcal{J}(x)&:=\text{diag}\left (\mathcal{J}_1(x_1)\cdots\mathcal{J}_k(x_k)\right ),\\g(x)&:=\text{diag}\left (g_1(x_1)\cdots g_k(x_k)\right ),\\H(x)&:=\sum_{j=1}^kH_j(x_j),\end{align}\end{subequations} where col($\cdot$) stacks the column vectors and diag($\cdot$) constructs a block-diagonal matrix. 
\end{definition}
\begin{definition}\label{QP_def}
   A \textit{quadratic program (QP)}  is a convex optimization problem with a quadratic cost function and linear constraints. It can be solved efficiently with any convex optimization tool.
\end{definition}
\begin{figure}[ht]
\centering\includegraphics[width=0.8\linewidth]{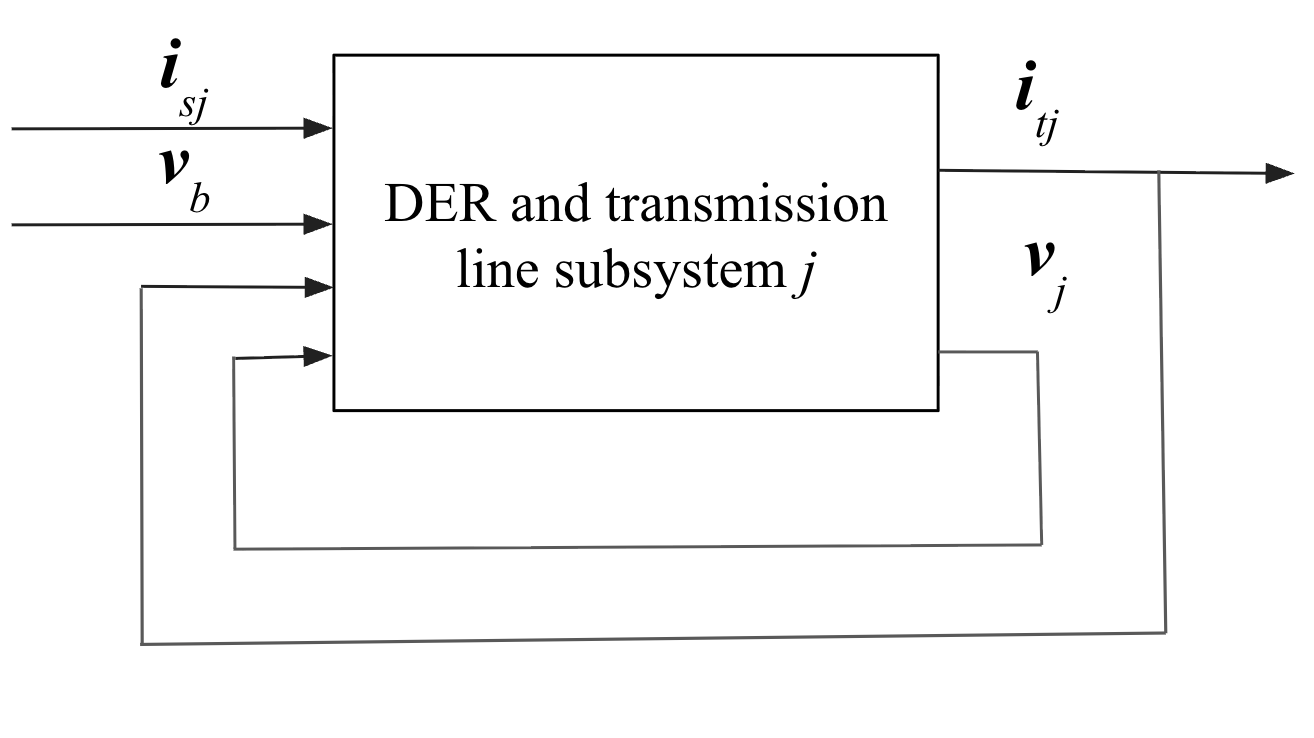}
    \caption{Input-output representation of DER and transmission line subsystems.}
    \label{fig:source subsystem}
\end{figure}
\begin{figure}[ht] \centering\includegraphics[width=0.8\linewidth]{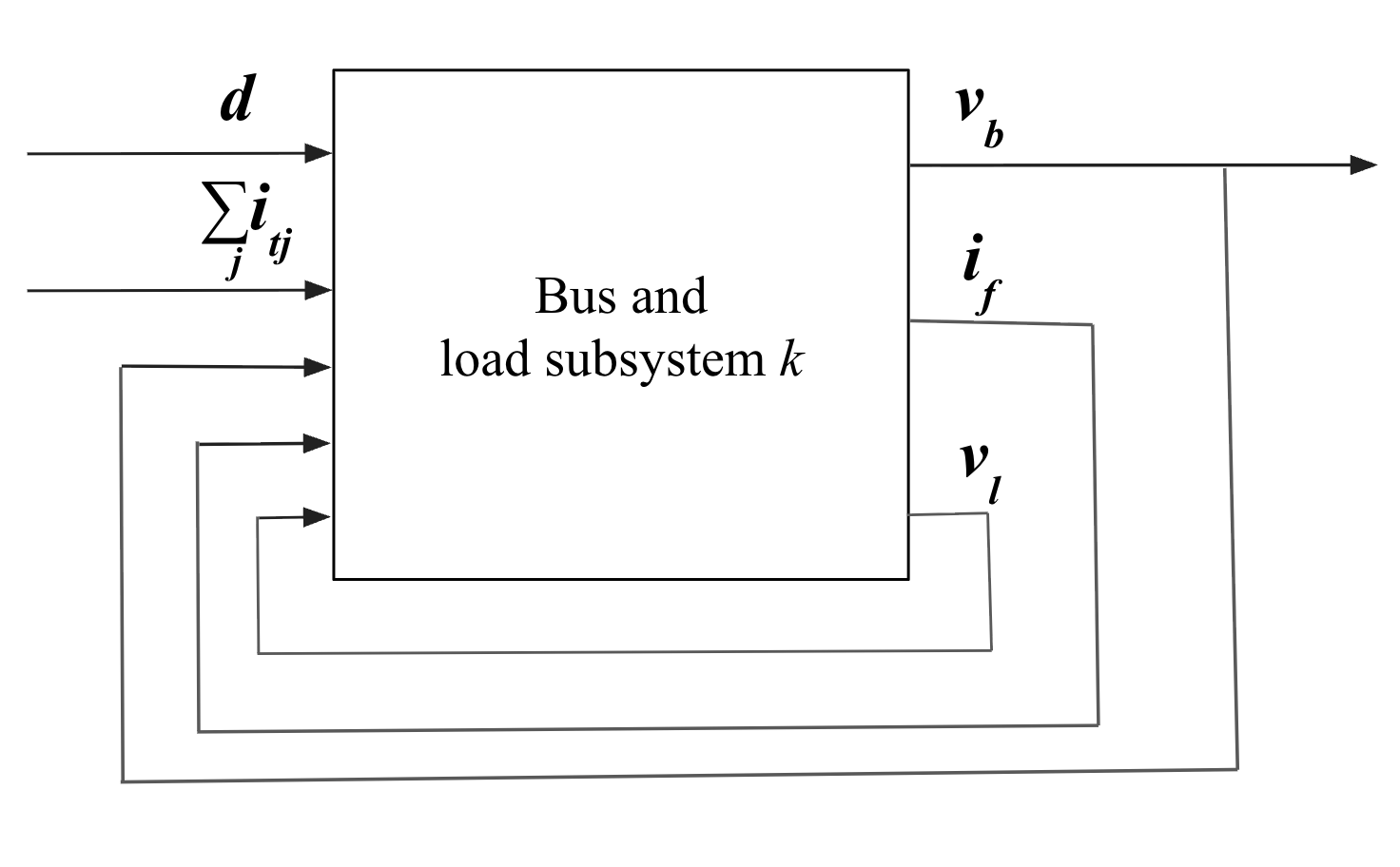}
    \caption{Input-output representation of the bus and load subsystem $k$}
    \label{fig:load subsystem}
\end{figure}
\begin{figure}[ht]
\centering
\includegraphics[width=0.48\textwidth]{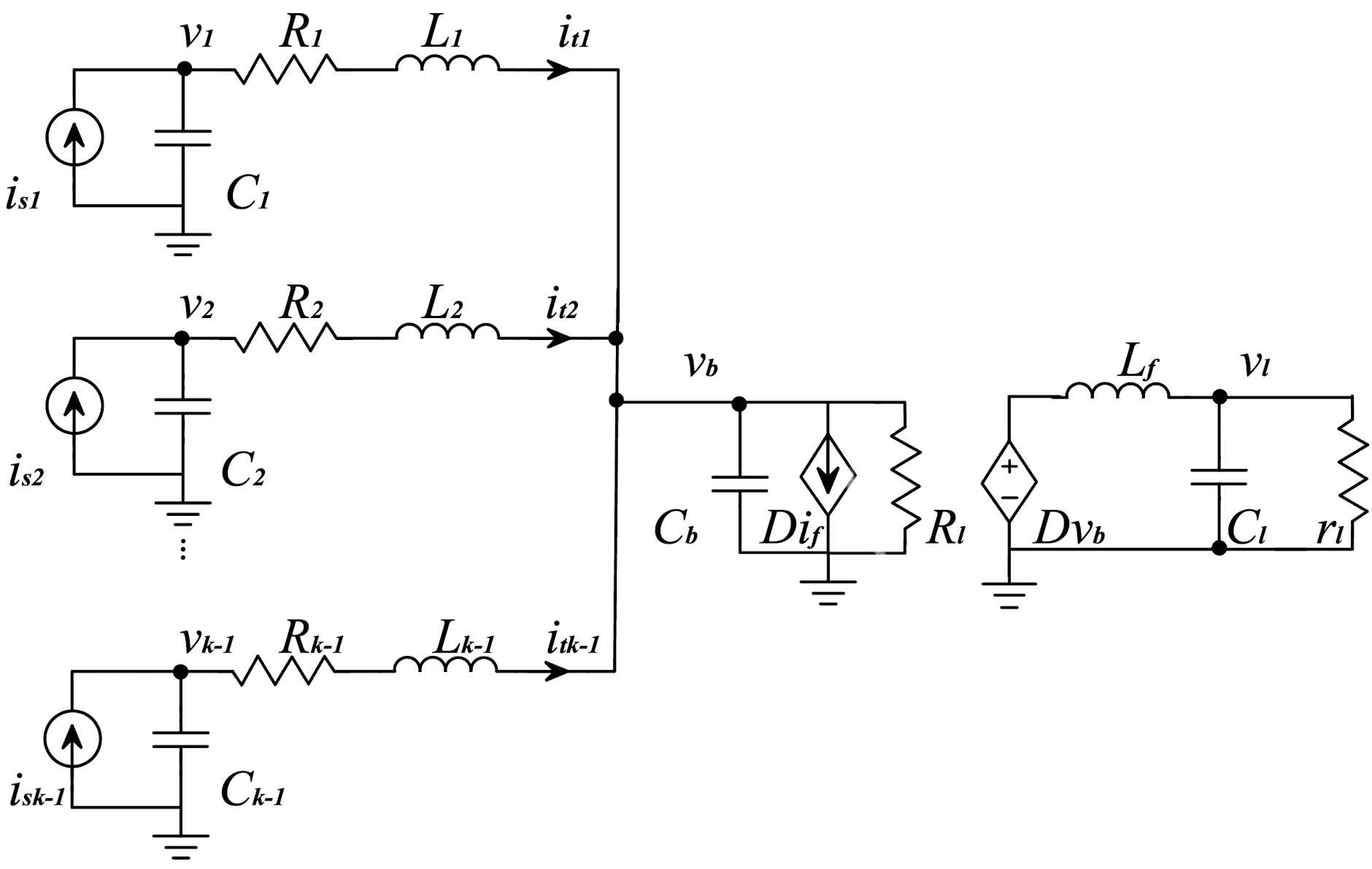}
\caption{Circuit schematics of a single-bus DC microgrid. A DER $j=1,\cdots,k-1$ is described by a current source $i_{s_j}$, its output capacitance $C_j$, and its output voltage $v_j$. A transmission line $j$ is given by resistance $R_j$, inductance $L_j$, and current through the line $i_{t_j}$. The bus is described by its voltage $v_b$ and capacitance $C_b$. The resistive load is $R_l$ and the DC load is given by its filter current $i_f$, load voltage $v_l$, and resistance $r_l$.}\label{fig:generic} 
\end{figure} 

\section{Modeling DC Microgrid Dynamics}
Consider a DC microgrid comprising $k-1$ DER that are connected to the grid through DC/DC converters with $k\geq1$. The following are assumed about the model \begin{enumerate}
    \item A DC microgrid is isolated from the utility grid and has a single bus that's radially connected to all DER;
    \item The transmission line impedance is time invariant.
\end{enumerate} 
To obtain a continuous-time dynamical model of the DC/DC converter, the true switched-circuit model is averaged over. On the power source side, a reduced model is employed to represent a DER and its
DC/DC converter that interconnects the DER with the microgrid. Specifically, the $j$-th DER and its converter are replaced by a current source and an output capacitance of the converter; see $i_{s_j}$ and $C_j$ in Fig. \ref{fig:generic}. A DC/DC converter in current mode control has an inner control loop that sets the inductor current to a reference value $i_{s_j}$ and an outer control loop that provides that value. The inner controller is on a much faster timescale than the outer controller which controls the averaged model. Hence, one can assume that $i_{s_j}$ is set instantaneously.

The load side of the microgrid is modeled as follows. All DC loads are aggregated into a pair of linear and nonlinear loads. The linear part is modeled by a resistor $R_l$, and the nonlinear part is modeled by a load-interfacing DC/DC converter. DC loads like consumer electronics are sensitive to voltage and power fluctuations. By controlling the turn ratio of the DC/DC converter, output voltage fluctuations are instantaneously suppressed.

The decentralized framework is more compatible with a model that utilizes all tunable turn ratios of the microgrid. Precisely, the DC/DC converter interfacing the load resistor can be modeled as a DC transformer with a turn ratio $0\leq d_l\leq1$ followed by a second-order low-pass filter. The filter inductor is $L_f$, the capacitor is $C_l$, and the load resistor is $r_l$. Lastly, the current through the filter inductor is $i_f$, and the filter's output voltage is $v_l$.

To fit the decentralized control framework, the microgrid is divided into $k$ interlinked subsystems. The dynamics of the $j$-th DER-interfacing converter and corresponding transmission line are combined into a subsystem $j$ with $j=1,\cdots,k-1$. Refer to Fig. \ref{fig:source subsystem} for the block-diagram representation of a source subsystem $j$. Here, the control input $i_{s_j}$ to the subsystem comes from a controlled current source, while the external input is the bus voltage $v_b$ that couples the subsystems with each other. The output of the subsystem is the current $i_{t_j}$ through the transmission line $j$. Applying KVL for loop $j$ in Fig. \ref{fig:generic} and KCL at node $v_j$, we arrive at
\begin{align}
     \dot{\begin{bmatrix}
    v_{j}\\
    i_{t_j}
\end{bmatrix}}&=\begin{bmatrix}
    -i_{t_j}/C_{j}\\
    (v_{j}-i_{t_j}R_{j}-v_b)/L_{j}
\end{bmatrix}+\begin{bmatrix} 1/{C_{j}}\\0
\end{bmatrix}i_{s_j}, 
    \label{eq:source}
\end{align}
where subsystems 1 through $k-1$ are coupled through the bus voltage $v_b$. 

The load subsystem $k$ combines the dynamics of the bus and the load-interfacing DC/DC converter. Refer to Fig. \ref{fig:load subsystem} for the block-diagram representation of a combined bus and load subsystem $j$. Here, the control input $d$ is the duty cycle of the DC/DC converter. The external coupling input is the sum of the transmission line currents $i_{t_j}$. The output of the subsystem is the bus voltage $v_b$. Applying KCL at nodes $v_b$ and $v_j$, and KVL for the loop $i_f$, the dynamics is found as 
\begin{align}
    \dot{\begin{bmatrix}
    v_b\\
    i_f\\
    v_l
\end{bmatrix}}=\begin{bmatrix}
    (i_{t_1}+\dots+i_{t_{k-1}}-v_b/R_l)/{C_b}\\
    -v_l/L_{f}\\
    (i_{f}-v_l/r_l)/C_f
\end{bmatrix}+\begin{bmatrix}
    -i_f/C_b\\
    v_b/L_f\\
    0
    \end{bmatrix}d_l.
    \label{eq:load_aff}
\end{align}

\section{Optimal Steady-State Behavior}\label{steady-state}
The natural steady-state values of DC microgrid dynamics are found by solving for the equilibrium of \eqref{eq:source} and \eqref{eq:load_aff}, but it results in an underdetermined system of equations. That leads to infinitely many steady-state values, parameterized by the steady-state current distribution, the bus voltage $v_b^*$ and the duty ratio $d_l^*$ for the load-interfacing converter. To find a unique and optimal set of steady-state values, a strictly convex optimal power flow problem is formulated. In particular, steady-state loss is minimized in the microgrid power network. The network loss is attributed to the steady-state power dissipated in the transmission lines due to their resistance $R_j$. The optimization problem is constrained with the steady-state dynamics of the bus \eqref{eq:load_aff}, and is formulated as follows
\begin{subequations} \label{eq:equil_constr}
    \begin{align} 
\min_{\{i_{t_j}^*\}_{j=1}^k}&\sum_{j=1}^k(i_{t_j}^*)^2R_{j}\\ 
        &\text{subject to }\sum_{j=1}^k i_{t_j}^*=\frac{v_b^*}{R_l}+\frac{d_l^*v_b^*}{r_l}.
    \end{align} 
\end{subequations} The solution to \eqref{eq:equil_constr} is given by \begin{equation}
     i_{t_j}^*=\frac{v_b^*/R_j}{\sum_{j=1}^n1/R_j}\left(\frac{1}{R_l}+\frac{d_l^*}{r_l}\right),\label{eq:solution_to10}\end{equation} for $j=1,\cdots,k.$ The optimal steady-state behavior given by \eqref{eq:solution_to10} ensures that the overall current supplied to the load is fairly distributed among the DER. The optimal current through the transmission is proportional to its conductance, which forces the DER to generate enough power to compensate for the dissipated loss. Hence, all DER effectively provide the same power to the load. The rest of the steady-state values are given by \begin{subequations}\label{eq:equil_all}\begin{align}v_j^*&=i_{t_j}^*R_j+v_b^*,\\
    i_{s_j}^*&=i_{t_j}^*,\\
    v_l^*&=d_l^*v_b^*,\\
 i_f^*&=v_l^*/r_l, \end{align}\end{subequations}
where $j=1,\cdots,k-1$.

\section{Decentralized Large Signal Stable Controller Design}\label{sec:stabilzing}\subsection{Port-Hamiltonian Representation of DC Microgrid}\label{subsec:GIP dynamics}
A nonlinear system that is port-Hamiltonian, see Def. \ref{def:PH}, comes with a natural Lyapunov function that certifies the large-signal stability of the system. Moreover, two or more PH systems can be easily interconnected through their external input and output ports, which is well suited to a decentralized control framework. Recall that the source dynamics in \eqref{eq:source} and the load dynamics in \eqref{eq:load_aff} are the natural interconnected subsystems of the microgrid. Then show that each subsystem is PH and rewrite the DC microgrid as a \textit{Global Interconnected Port-Hamiltonian (GIPH)} system described in Def. \ref{def:global_pH}.  

Consistent with Section III, there are $k$ subsystems comprising the DC microgrid. Subsystems 1 through $k-1$ describe the DER and its converter \eqref{eq:source}. The following PH subsystem is equivalent to \eqref{eq:source}
\begin{subequations}\label{eq:PH_source}
    \begin{align}
        x_j&:=\begin{bmatrix}
    v_j &i_{t_j}
\end{bmatrix}^T,\\
H_j(x_j)&:=(C_jv_j^2+L_ji_{t_j}^2)/2,\\
\Rightarrow\partial H_j(x_j)&:=\begin{bmatrix}
    C_jv_j&L_ji_{t_j}
\end{bmatrix}^T,\\
\mathcal{J}_j(x_j)&:=\begin{bmatrix}
    0 &-1/L_jC_j\\
    1/L_jC_j& 0
\end{bmatrix},\\
\mathcal{R}_j(x_j)&:=\begin{bmatrix}
    0 &0\\
    0&R_j/L_j^2
\end{bmatrix};\\
u_j&:=i_{s_j}\Rightarrow g_j(x_j):=\begin{bmatrix}
    1/C_j& 0
\end{bmatrix}^T;\\
z_j&:=-v_b\Rightarrow g_{z_j}(x_j):=\begin{bmatrix}
    0&1/L_j
\end{bmatrix}^T;\\y_j&:=i_{t_j},\end{align}\end{subequations} where $j=1,\cdots,k-1.$ 

The subsystem $k$ of the microgrid represents the DC bus and the DC load dynamics \eqref{eq:load_aff}.  Due to the difference between \eqref{eq:source} and \eqref{eq:load_aff}, the load side is rewritten as the following PH subsystem
\begin{subequations}\label{eq:PH_load}
    \begin{align}
        x_k&:=\begin{bmatrix}
    v_b &i_f&v_l 
\end{bmatrix}^T,\\
H_k(x_k)&:=(C_bv_b^2+L_fi_{f}^2+C_fv_l^2)/2,\\\Rightarrow\partial H_k(x_k)&:=\begin{bmatrix}
    C_bv_b&L_fi_{f}&C_fv_l
\end{bmatrix}^T,\\
\mathcal{J}_k(x_k)&:=\begin{bmatrix}
    0&0&0\\0 &0&-1/L_fC_f\\
    0&1/L_fC_f& 0
\end{bmatrix},\\
\mathcal{R}_k(x_k)&:=\begin{bmatrix}
    1/R_lC_b^2&0 &0\\0&0&0\\0&0&1/r_lC_f^2
\end{bmatrix};\\
u_k&:=d_l,\\\Rightarrow g_k(x_k)&:=\begin{bmatrix}
    -i_f/C_f &v_b/L_f&0
\end{bmatrix}^T;\\
z_k&:=-\begin{bmatrix}
    i_{t_1}&\cdots&i_{t_{k-1}}
\end{bmatrix}^T,\\\Rightarrow g_{z_k}(x_k)&:=-\begin{bmatrix}
    1/C_b&0&0\\1/C_b&0&0
\end{bmatrix}^T;\\y_k&:=-\begin{bmatrix}
    v_b&\cdots&v_b
\end{bmatrix}^T\in\mathbb{R}_{\leq0}^{k-1}.
    \end{align}
\end{subequations}

Lastly, note that \eqref{eq:PH_source} and \eqref{eq:PH_load} are coupled through their interconnected input and output ports. To verify that the PH subsystems make up a GIPH, check if there is any external port left 
\begin{align*}\sum_{j=1}^ky_j^Tz_j&=\sum_{j=1}^{k-1}-i_{t_j}v_b+v_bi_{t_j}=0.\end{align*}\subsection{Large-Signal Stable DC Microgrid}\label{subsec:stableGIPH}Recall the desired steady-state values for the terminal voltage $v_{j}^*$ of $j$-th DER-interfacing converter and the current $i_{t_j}^*$ through $j$-th transmission line from \eqref{eq:equil_all} and \eqref{eq:solution_to10}, respectively. To achieve an exponential decay to equilibrium $x_j^*$, see Def. \ref{def:global_pH}, the open-loop dissipation matrix $\mathcal{R}_k$ \eqref{eq:PH_source} lacks a dissipative term for the terminal voltage such as $-\alpha_jv_j$ with $\alpha_j>0$ in the first diagonal entry. With that in mind, set the closed-loop dynamics for each source PH subsystem $j$ to be a linear time-invariant (LTI) system with an exponentially stable equilibrium $x_j^*$ by defining
\begin{subequations}\label{eq:PH_source_desired}
    \begin{align}
        \hat{x}_j&:=x_j-x_{j}^*,\\
        \mathcal{J}_{j}^*(x_j)&:=\begin{bmatrix}
    0 &-1/L_jC_j\\
    1/L_jC_j& 0
\end{bmatrix},\\
\mathcal{R}_{j}^*(x_j)&:=\begin{bmatrix}
    \alpha_j/C_j^2 &0\\
    0&R_j/L_j^2
\end{bmatrix},\\
\hat{u}_j&:=u_j-u_j^*,\\
\hat{z}_j&:=z_j-z_j^*,\\
\hat{y}_j&:=y_j-y_j^*,\end{align}\end{subequations}
where $j=1,\cdots,k-1$. 
The closed-loop dissipation matrix is given by $\mathcal{R}^*_j$ and the interconnection matrix is $\mathcal{J}_j^*$. Lastly, $u_j^*,z_j^*$ and $y_j^*$ are desired steady-state values for control, external input, and output of the source PH subsystem.

Switching to the load PH subsystem defined in \eqref{eq:PH_load}, its desired closed-loop behavior should also be an LTI system with an exponentially stable $x_k^*$. Note that open-loop load dynamics lacks a dissipation term in $\mathcal{R}_k$ for the filter current $i_f$. To address this in closed-loop, replace 0 in the second diagonal entry of $-\mathcal{R}_k^*$ with$-\alpha_k<0$. The steady-state control input $-d_l^*\hat{i}_{tf}$ is added to the bus dynamics $\dot{v}_b$ and $d_l^*\hat{v}_b$ to $\dot{i}_f$. This results in addition of closed-loop entries $-d^*_l/C_bL_f$ and $d^*_l/C_bL_f$ to entries 2 by 2 and 3 by 3 of the interconnection matrix $\mathcal{J}_j^*$. Then the desired closed-loop dynamics is 
\begin{subequations}\label{eq:PH_load_desired}
    \begin{align}
        \hat{x}_k&:=x_k-x^*_k,\\
        \mathcal{J}_{k}^*(x_k)&:=\begin{bmatrix}
    0 &d_l^*/C_bL_f&0\\
    -d_l^*/C_bL_f& 0&-1/C_fL_f\\
    0&1/C_fL_f&0
\end{bmatrix},\\
\mathcal{R}_{k}^*(x_k)&:=\begin{bmatrix}
    1/R_lC_b^2 &0&0\\
    0&\alpha_k/L_f^2&0\\
    0&0&1/r_lC_f^2
\end{bmatrix},\\
\hat{d}_l&:=d_l-d_l^*.\end{align}\end{subequations}
\subsection{Nominal Stabilizing Controller}
Recently proposed Dynamic Interconnection and Damping Assignment Passivity-Based Control (IDA-PBC) approach \cite{yuan2025large} is followed to design the nominal stabilizing controller. Traditional IDA-PBC has two key limitations. First, it does not establish the port-passivity of each PH subsystem, which compromises the stability of GIPH. Second, the resulting controller is singular at the desired equilibrium $x_j^*$, which breaks Lipschitz continuity. The dynamic IDA-PBC solves for the control $u_j$ and the dynamic state $\tilde{x}_j$ in the following equation 
\begin{align}
    \left[\mathcal{J}^*_j-\mathcal{R}^*_j\right]&\partial H_j(\hat{x}_j)-g_{z_j}z^*_j+\dot{\tilde{x}}_j\nonumber\\&=\left[\mathcal{J}_j-\mathcal{R}_j\right]\partial H_j(x_j)+g_j(x_j)u_j,\label{eq:ida_pbc}
\end{align}
where $j=1,\cdots,k$. The dependence on $x_j$ is removed from the constant matrices $\mathcal{J}_j,\mathcal{R}_j,g_{z_j}$ and $\mathcal{J}_j^*,\mathcal{R}_j^*$. The resulting control law for the source subsystem $j$ is given by
\begin{equation}
    \hat{u}_j(\hat{x}_j)=-\alpha_j\hat{v}_j+e^{-R_jt/L_j}\hat{i}_{t_j}(0),\label{eq:u_j_hat}
\end{equation} where $j=1,\cdots,k-1$. The control law for the load subsystem $k$ is given by
\begin{equation}
    \hat{d}_l(\hat{x}_k)=-\begin{cases}
        \alpha_k\hat{i}_f/v_b, v_b>0; \\
        0, v_b=0.
    \end{cases}\label{eq:d_l_hat}
\end{equation}

\section{Decentralized Safety-Critical Controller (DSCC) Design}
The nominal stabilizing controller in Section \ref{sec:stabilzing} does not come with any safety guarantees. That is particularly concerning since the microgrid can fail completely due to hazardous currents and voltages in the power network circuitry. Unsafe conditions can be caused by large-signal disturbances, such as cyber attacks. This work targets the safety adherence of the microgrid and propose the Decentralized Safety-Critical Controller (DSCC). The DSCC is formulated as an optimal control problem that combines safety and stability design objectives to solve for a unique feedback controller. 
\subsection{Stability by Control Lyapunov Function}
Let the global Hamiltonian $H$ be a candidate Lyapunov function that certifies the global exponential stability of the desired equilibrium $x^*$ for the closed-loop DC microgrid dynamics. To see that $H$ is positive definite and differentiable, consider
\begin{subequations}
    \begin{align}
        &H(\hat{x})=0\iff x=x^*,\\&H(\hat{x})=\hat{x}^TQ\hat{x}/2,\\
        &Q=\text{diag}(C_1,L_1,\cdots,C_{k-1},L_{k-1}C_b,L_f,C_f)>0\\
        &\partial H(\hat{x})=Q\hat{x}.
    \end{align}\label{eq:H_CLF}
\end{subequations}
Under the stabilizing controller, the closed-loop PH subsystems are given by \eqref{eq:PH_source_desired} and \eqref{eq:PH_load_desired}. The derivative of each subsystem Hamiltonian $H_j$ along the closed-loop DC microgrid dynamics:
\begin{align}
     \dot{H}_j(\hat{x}_j)&=\partial H_j(\hat{x}_j)^T\left [\mathcal{J}^*_j-\mathcal{R}^*_j \right ]\partial H_j(\hat{x}_j)+\hat{y}_j^T\hat{z}_j,\nonumber\\
     &=-\partial H_j(\hat{x}_j)^T\mathcal{R}^*_j\partial H_j(\hat{x}_j)+\hat{y}_j^T\hat{z}_j.\label{eq:subHamiltonian}
\end{align} where $j=,\cdots,k$. By Def. \ref{def:global_pH}, the derivative of the global Hamiltonian is the sum of \eqref{eq:subHamiltonian}:
\begin{align}
    \dot{H}(\hat{x})=\sum_{j=1}^k\dot{H}_j(\hat{x}_j)&=-\sum_{j=1}^k\partial H_j(\hat{x}_j)^T\mathcal{R}^*_j\partial H_j(\hat{x}_j),\nonumber\\
    &=\partial H(\hat{x})^T\mathcal{R}
^*\partial H(\hat{x})<0,
    \label{eq:globHamiltonian}
\end{align}
and \eqref{eq:globHamiltonian} establishes the global exponential stability of $x^*$ since no restrictions were placed on the state-space of $\hat{x}$.
$H(\hat{x})$ is a candidate \textit{Control Lyapunov Function} (CLF) given by Def. \ref{def:CLF} for the DC microgrid dynamics. To show that it is indeed a CLF,  define
\begin{subequations}\label{eq:PH_CLF}
    \begin{align}&f_j(x_j):=(\mathcal{J}_j-\mathcal{R}_j)\partial H_j(x_j),\\&f(x):=\text{col}(f_1(x_1)\cdots f_k(x_k))\\
&\hat{\mathcal{R}}^*:=\mathcal{R}^*-\Lambda>0\\
     &\mathcal{L}_fH(\hat{x})+\mathcal{L}_gH(\hat{x})u\leq-\hat{x}^TQ\hat{\mathcal{R}}^*Q\hat{x}
    \end{align}
\end{subequations}
  where the positive diagonal entries of the diagonal matrix $\Lambda$ of dimension $2k+1$ are tunable. Then, check
\begin{align*}\mathcal{L}_gH(\hat{x})=0&\Leftrightarrow \hat{v}_j=0, v_bi_f^*=v_b^*i_f\\
     &\Rightarrow \mathcal{L}_fH(\hat{x})+\hat{x}^TQ\hat{\mathcal{R}}^*Q\hat{x}\\&=-\sum_{j=1}^{k-1}\Lambda_{2j}\hat{i}^2_{t_j}-\Lambda_{2k-1}\hat{v}_b^2-\Lambda_{2k+1}\hat{v}^2_l\\
     &<0.\end{align*}
Hence, the global Hamiltonian $H(\hat{x})$ is a CLF for the global system, and the desired equilibrium $x^*$ can be made exponentially stable if $u$ is picked such that \eqref{eq:PH_CLF} holds for all $\hat{x}$.
Locally, if each Hamiltonian $H_j(\hat{x}_j)$ satisfies
\begin{equation}\mathcal{L}_{g_j}H_j(\hat{x}_j)u_j+\mathcal{L}_{f_j}H_j(\hat{x}_j)\leq\hat{y}_j^Tz_j^*-\hat{x}_j^TQ_j\hat{\mathcal{R}}_jQ_j\hat{x}_j,\label{eq:lemma1_u_j}\end{equation} then the global $x^*$ can be made exponentially stable. Notice that all the states are local and there is no coupling between the controllers, since the constraint only requires the knowledge of the desired external input $z^*$. As a result, the controller can be implemented in a  \textit{decentralized manner}, without sharing the state measurements of neighboring or global state measurements. 
\subsection{Safety by Control Barrier Functions}

An unexpected cyberphysical attack targeting a control input $u_j$ can quickly destabilize the entire DC microgrid causing violation of safety constraints. The safety constraints are there to prevent system failures and permanent damage due to physical limitations of hardware components. More importantly, voltage and current fluctuations across the microgrid pose a huge risk to the health and safety of people relying on the power network. One of the most common ways to define safety restrictions for smooth operation of the DC microgrid is to specify the maximum and minimum safety thresholds. Following that, consider the safe sets defined as 
\begin{subequations}\label{eq:safe_set_v}\begin{align}&\mathcal{C}_j:=\{x_j|v_{j,\min}\leq v_j\leq v_{j,\max}\}, j=1,\dots,k-1,\\&\mathcal{C}_k:=\{x_k|i_{f,\min}\leq i_f\leq i_{f,\max}\},\\&\mathcal{C}_0:=\cap_{j=1}^k\mathcal{C}_j.\end{align}\end{subequations}
where $v_{j,\max}$ and $v_{j,\min}$ refer to the upper and lower safety limits for the output voltage $v_j$ of the DER-interfacing converter $j$. Similarly, $i_{f,\min}$ and $i_{f,\max}$ are the safety limits for current $i_f$ through the inductor of the second-order low-pass filter preceding the load resistor $r_l$. 

Recently \textit{Control Barrier Functions}\cite{taylor2023robust} were validated experimentally as a viable way to generate controllers that ensure the state trajectories remain in the safe region. CBFs generalize the concept of controlling the overshoot of the state trajectories under a stabilizing controller for an LTI system. CBFs are well equipped to handle cyberphysical attacks that can be represented as impulse disturbances causing the microgrid to steer away from its equilibrium, which are often caused by cyber attacks such as DoS and FDI.

Consider now a candidate CBF for the safe set $\mathcal{C}_j$ under the dynamics of the source
PH subsystems $j$
    \begin{equation}
     B_j(x_j):=-\frac{1}{(v_{j}-v_{j,\min})(v_{j}-v_{j,\max})},\label{eq:def_cbd}
    \end{equation} where $j=1,\cdots,k-1$. To confirm that they are CBFs, recall Def. \ref{def:CBF}, and note that 
    \begin{align*}
        \mathcal{L}_{g_j}B_j=0&\Leftrightarrow \partial B_j=0,\\
        &\Leftrightarrow v_j^0:=(v_{j,\min}+v_{j,\max})/2,\\&\Rightarrow B_j(v_j^0)>0,\text{ but }\partial B_j=0\\
        &\Rightarrow \mathcal{L}_{f_j}B_j=0<\beta_j/B_j(v_j^0).
    \end{align*}
    Consider now the safe set $\mathcal{C}_k$ for the current through the filter inductor $i_f$. Then, a candidate CBF for the load PH subsystem $k$ is given by
        \begin{equation}
     B_k(x_k):=-\frac{1}{(i_{f}-i_{f,\min})(i_{f}-i_{f,\max})}\text{ with } v_b>0.\label{eq:def_cbd_load}
    \end{equation} To confirm $B_k(i_f)$ is a CBF,
    \begin{align*}
        \mathcal{L}_{g_k}B_k=0&\Leftrightarrow i_f^0:=(i_{f,\min}+i_{j,\max})/2,\\
         &\Rightarrow B_k(i_f^0)>0,\Rightarrow0<\beta_k/B_k(i_f^0.
    \end{align*}
    By the definition of CBF, see Def. \ref{def:CBF}, safe controllers render the interior of the safe set $\text{int}(\mathcal{C})$ forward-invariant, and can be described with the following inequalities 
\begin{align}\mathcal{L}_{f_j}B_j(x_j)+\mathcal{L}_{g_j}B_j(x_j)u_j\leq \frac{\beta_j}{B_j(x_j)}, \forall x_j\in\text{int}(C_j),\label{eq:cdf_constraint}
\end{align} 
for each $j=1,\dots,k$. Importantly, \eqref{eq:cdf_constraint} are linear and separable in $u_j$. In addition, the Lie derivatives $\mathcal{L}_{f_j}B_i(x_j)$ and $\mathcal{L}_{g_j}B_j(x_j)$, and the CBFs $B_j(x_j)$ are all functions of local states $x_j$. These seamlessly fit the decentralized control architecture proposed in the current work. 
\subsection{Safe-Stabilization as Optimal Control Problem}
Finally, our main results are stated here. To start, Theorem 1 from \cite{abdirash2025nonlinearoptimalcontroldc} is restated which provides safety and stability guarantees for a centralized safety-critical controller controller from \cite{abdirash2025nonlinearoptimalcontroldc}. Following that, Theorem 2 is proposed to establish the safety and stability guarantees of the DSCC using the results of Theorem 1.
\begin{theorem}\label{jankovic}
Let $V (x)$ be a CLF and $B(x)$ be a CBF for an open-loop system in \eqref{eq:aff}. If a feedback control law $k:X\rightarrow U$ for a closed-loop system in \eqref{eq:closedloop} is a solution to the following QP (Def \ref{QP_def})
\begin{align}
        &\min_{(u, \delta)\in U\times\mathbb{R}^{n}}||u-u_{0}(x)||^2+m||\delta||^2\nonumber\\
        \text{subject to }&\gamma(\mathcal{L}_{f}V(x)+\alpha||x||^2)\nonumber\\&+\mathcal{L}_{g}V(x)\left(u+\delta\right) \leq\ 0\nonumber,\\&\mathcal{L}_fB(x)+\mathcal{L}_gB(x)u-\frac{\beta}{B(x)}\leq0,
        \label{qp}
\end{align} 
where $\gamma(p):=\begin{cases}
    \frac{m+1}{m}p, p\geq0\\
     p, p<0,
\end{cases}$ is a correction factor for the Lyapunov decay condition, free variable $m>0$ is the penalty for a slack variable $\delta$, and  $u_{0}(x)$ is a nominal Lipschitz continuous and stabilizing controller. Then,
\begin{enumerate}
    \item The QP problem is feasible for all $x\in X$ and the resulting control law $u^*(x)$ is Lipschitz continuous of order $m$ in every subset of $\text{int}(\mathcal{C})$ not containing the equilibrium.
    \item $\dot{B}(x)=\mathcal{L}_fB(x)+\mathcal{L}_gB(x)u\leq\frac{\beta}{B(x)}$ for all $x\in\text{int}(\mathcal{C})$ and the set $\text{int}(\mathcal{C})$ is forward invariant.
    \item If the barrier constraint is inactive and $\gamma m = 1$ is satisfied, the control law achieves $\dot{V}_\eta(x)=\mathcal{L}_{f}V(x)+\mathcal{L}_{g}V(x)u\leq-\alpha||x||^2.$
    \item If $x^*\in\text{int}(\mathcal{C})$ and the CLF $V$ has the CCP, then the barrier constraint is inactive around $x^*$, $u^*(x)$ is continuous at $x^*$, and the closed loop system is locally exponentially stable at $x^*$.
\end{enumerate}
\end{theorem}
\begin{theorem}\label{ours}
    Let $H_j$ be the Hamiltonian for the $j$-th interconnected PH subsystem of the DC microgrid. Similarly, let $B_j$ be the CBF for safe set $\mathcal{C}_j$. Let $u_{jPH}$ be the nominal stabilizing controller given by \eqref{eq:u_j_hat} and \eqref{eq:d_l_hat}. Then, a local feedback controller $u_j$ can be formulated as a solution to the following QP
\begin{align}u_j^*(x_j)=&\argmin_{(u_j,\delta_j)\in\mathbb{R}^{2}}
            ||u_j-u_{jPH}(x_j)||^2+m_j||\delta_j||^2\nonumber\\
        &\text{subject to }\nonumber\\&\hspace{0cm}\gamma_j\left(\mathcal{L}_{f_j}H_j(\hat{x}_j)+\alpha_j||\hat{x}_j||^2-\hat{y}_j^Tz_j^* \right)\nonumber\\&+\mathcal{L}_{g_j}H_j(\hat{x}_j)\left(u_j+\delta_j\right)\leq\ 0\nonumber,\\&\hspace{0cm}\displaystyle \mathcal{L}_{f_j}B_j(x_j)+\mathcal{L}_{g_j}B_j(x_j)u_j\leq \frac{\beta_j}{B_j(x_j)},\label{eq:qp2}\end{align}
where $j=1,\dots,k,$, $\gamma_j(p):=\begin{cases}
    \frac{m_j+1}{m_j}p, p\geq0\\
    p, p<0,
\end{cases}$, $m_j>0$ and $\alpha_j$ is the minimum eigenvalue of $\hat{\mathcal{R}}^*_j$. Then the local controller $u_j^8(x_j)$ given by \eqref{eq:qp2} can be implemented in real-time to control the DC microgrid in a decentralized fashion. If so, it stabilizes the desired equilibrium of each subsystem, which implies the stability of the global equilibrium. Importantly, it ensures that (i) the terminal voltages of the converters, interfacing the DER, and (ii) the filter current of the converter, interfacing the load resistor, stay within their specified maximum and minimum safety limits.
\end{theorem}
\begin{proof}
    Earlier in this section, we showed that $H$ is a global CLF in \eqref{eq:globHamiltonian}
    and $\{B_j\}_{j=1}^k$ are CBFs in \eqref{eq:def_cbd} and \eqref{eq:def_cbd_load}.  Set $\alpha=\min_j\alpha_j$, and notice that each CBF is separable in both $u_j$ and $x_j$, resulting in the CBF conditions being independent of each other. Then, applying Theorem \ref{jankovic} to each pair of $H$ and $B_j$ and solving for $u^* _j$, the results follow.
\end{proof}

\section{Validation and Discussion}\label{numerical}
Consider a single-bus DC microgrid with two DEGs supplying an aggregated nonlinear load. The circuit schematic of the DC microgrid is shown in Fig. \ref{fig:plecs_schem}. The model parameters are given in Table \ref{tab:my_label1}. Our objectives are as follows
\begin{enumerate}
    \item keeping the terminal voltage $v_j$ of each DC/DC converter within its safety bounds; 
    \item keeping switch the current, equivalently $i_f, i_{s1}, i_{s2}$, within the safety limits; 
    \item regulating the bus voltage $v_b$ to its optimal steady-state value $v_b^*$. 
\end{enumerate}
The optimal steady-state values of the microgrid are listed in Table \ref{tab:my_label}. The simulation results are presented in Fig. \ref{fig:plecs_validation}, which compares the performance of the proposed DSCC with the Dynamic IDA-PBC of \cite{yuan2025large}. Note that a Zener diode was used to limit the terminal voltage of a converter $v_j$, and a current limiter was used to limit the filter inductor current $i_f$. Design parameters of the DSCC tuned to the given DC microgrid are listed below.
\begin{enumerate}
    \item $m_{1,2,3}=10$ to prioritize safety in the trade-off between safety and stability, and limit the slew rate of the current sources;
    \item CBF decay rates of $\beta_{1,2,3}=0.1$ so that the states approach the safe boundary slow enough to prevent safety violation due to sampling and quantization errors;
    \item CLF decay rates of $\alpha_3=0.6>\alpha_{1,2}=0.01$ to make sure the load subsystem reaches its operating point faster than the source subsystems.
\end{enumerate} The steady-state values achieved by the DSCC and the Dynamic IDA-PBC are reported in the last two columns of Table \ref{tab:my_label}, respectively. The results demonstrate that the proposed DSCC effectively stabilizes the system around $x^*$ and adheres to the safety constraints, despite being initialized far from the equilibrium. In contrast, the Dynamic IDA-PBC struggles to maintain stability and violates the safety constraints.
\begin{figure*}[ht]
    \centering
\includegraphics[width=0.76\textwidth]{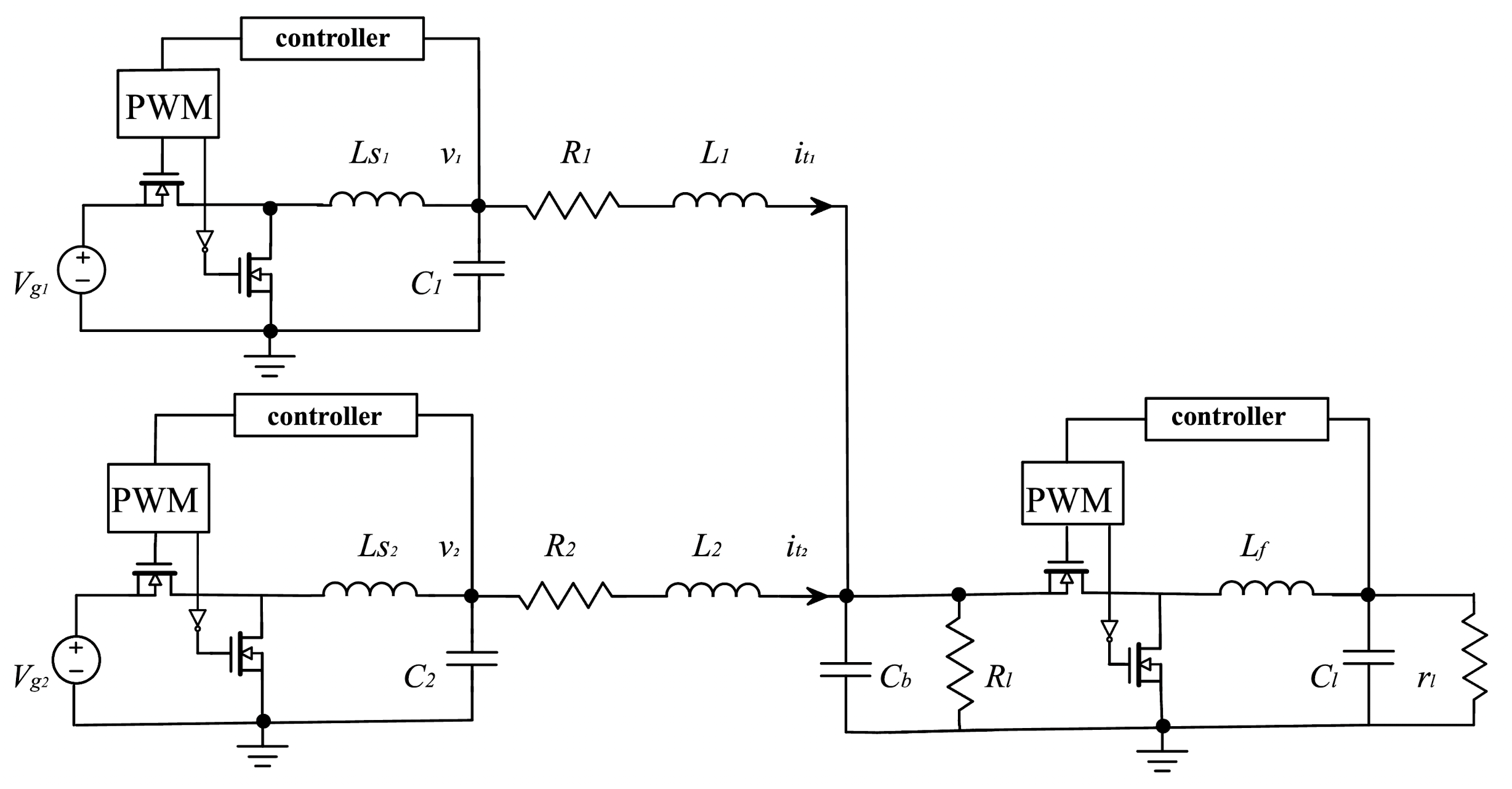}
    \caption{Switched-circuit model of the DC microgrid which was used in high-fidelity simulation validation of the proposed DSCC.}
    \label{fig:plecs_schem}
\end{figure*}
\setlength{\tabcolsep}{2pt}
\begin{table}[ht]
    \centering
    \caption{Model parameters of DC microgrid circuitry in Fig. \ref{fig:plecs_schem} for the numerical experiment. Note that each filter inductor has a parasitic series resistance. The load consumes $P=1463$ W and the switching frequency is $f_s=200$ kHz.}\small
    \begin{tabular}{c|c|c|c|c|c|c}
       $j$&$V_{g_j}$ [V] & $ L_{s_j}$  [mH] & $r_{s_j}$  [m$\Omega$] & $   C_j $  [mF] & $   L_j$  [mH] & $   R_j$  [m$\Omega$] 
       \\
     \hline
      
         1& 48& 0.159& 3.552&  0.09 &   0.49 &   18.78\\ 2& 48& 0.159& 3.552 &   0.07 &   0.48 &   17.78\\ \hline
           Load & &$L_f$ [mH]  & $r_f$ ]m$\Omega$]& $C_b$ [mF] &    $r_l$ [m$\Omega$]  &    $R_l$ [m$\Omega$]\\\hline
     & &  0.16  & 2& 0.47 &  175   &   2000  \\
    \end{tabular}
    \label{tab:my_label1}
\end{table} 
\begin{figure*}[ht]
\centering\includegraphics[width=0.65\linewidth]{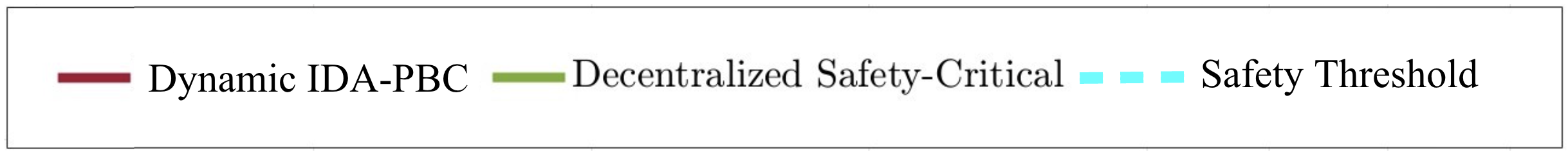}
\begin{multicols}{3}
\includegraphics[trim=0.5cm 0 0 0.5cm,clip=true,width=\linewidth]{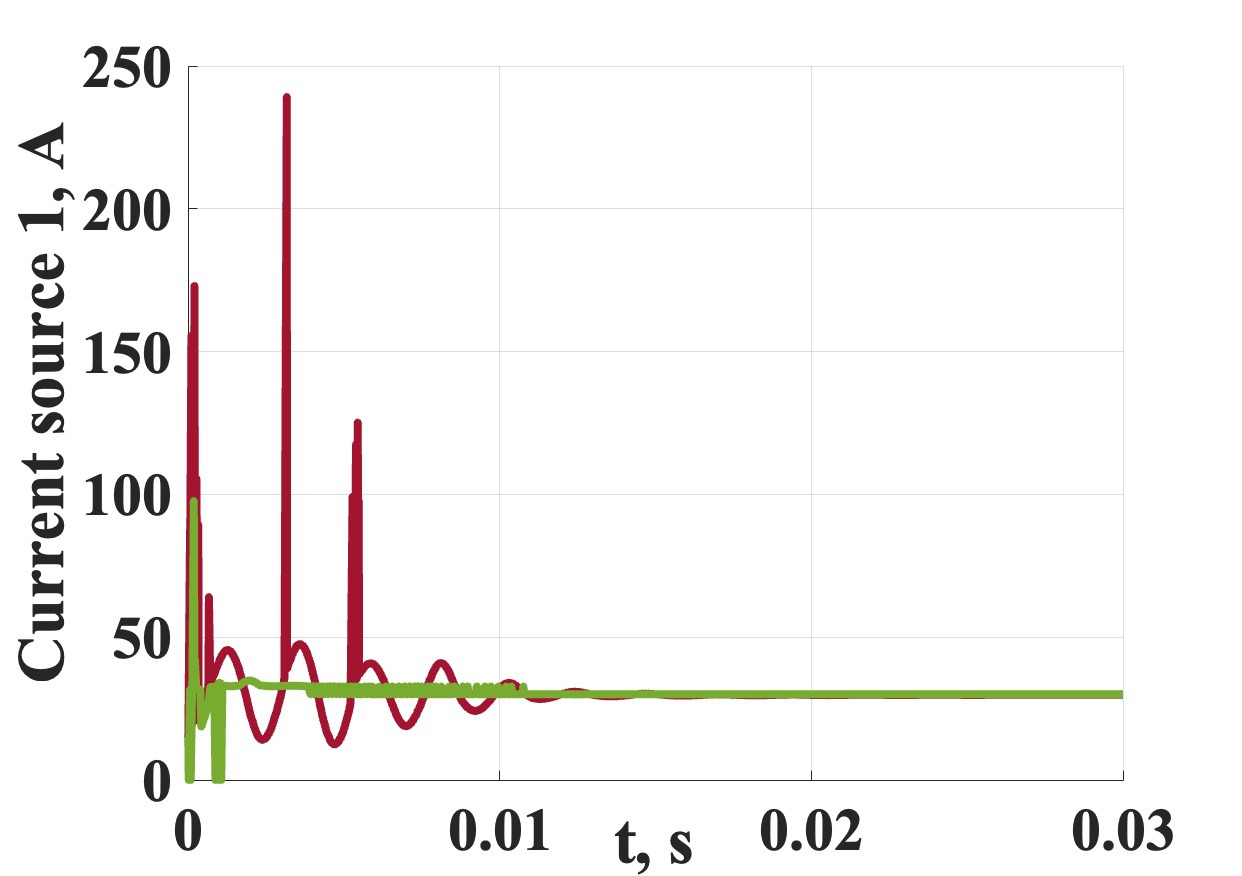}\par
\includegraphics[width=\linewidth]{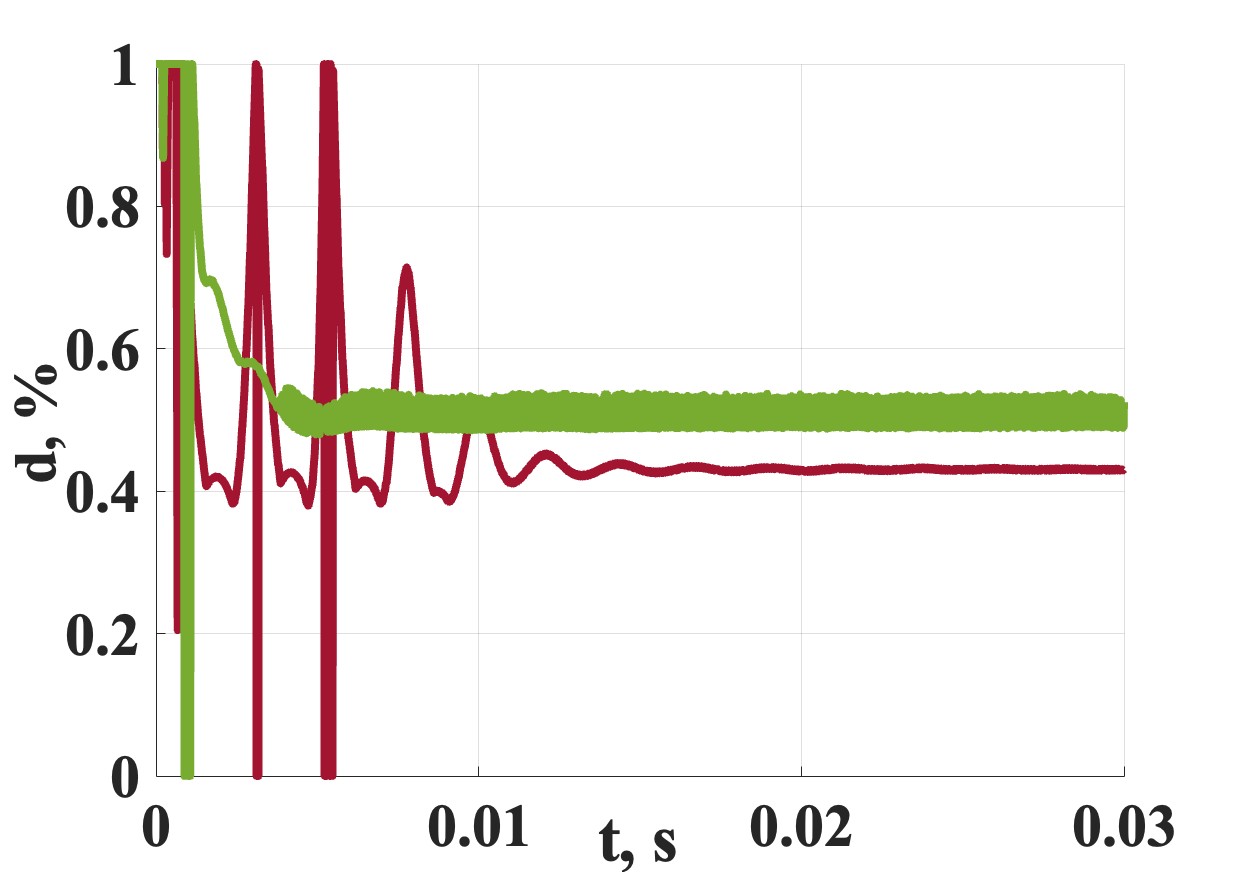}\par
\includegraphics[width=\linewidth]{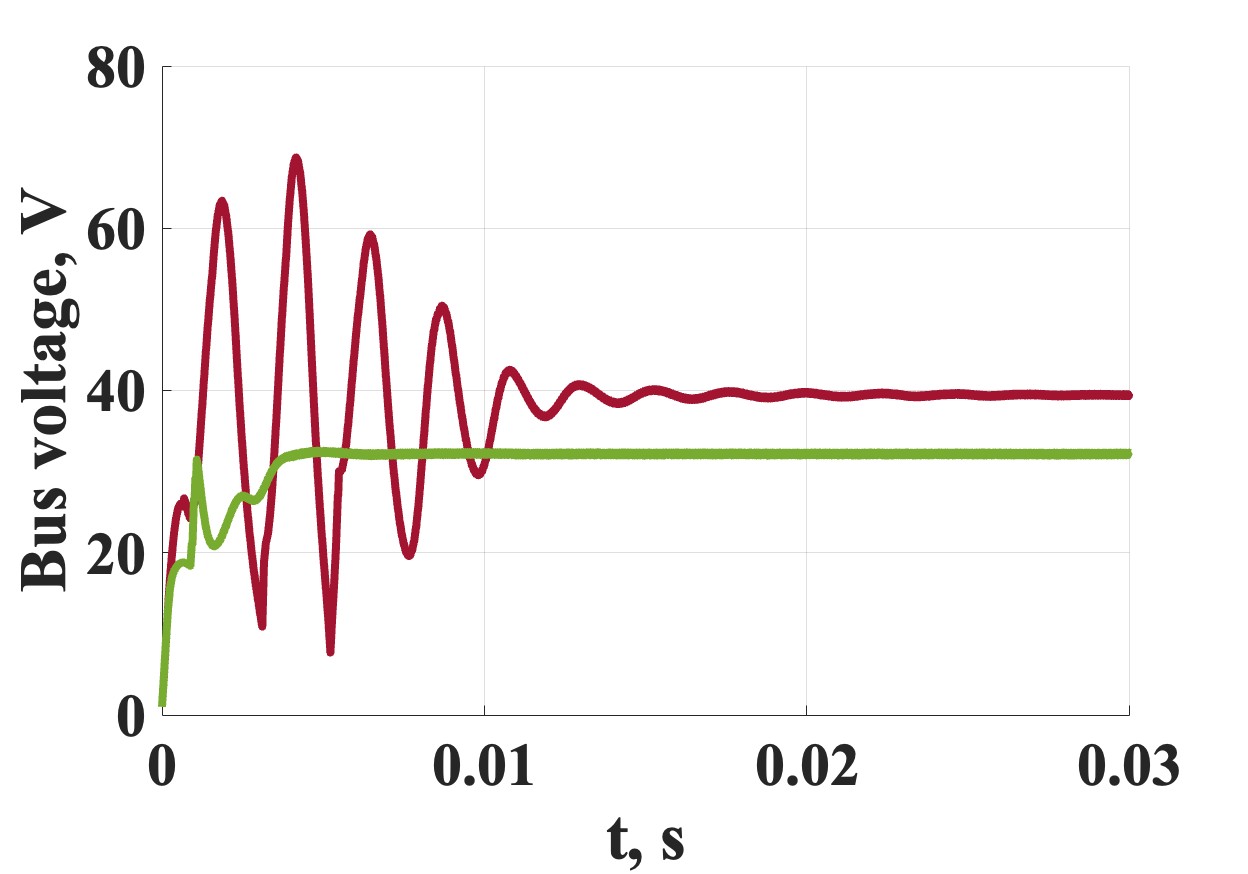}\par
\end{multicols}
\begin{multicols}{3}
 \includegraphics[width=\linewidth]{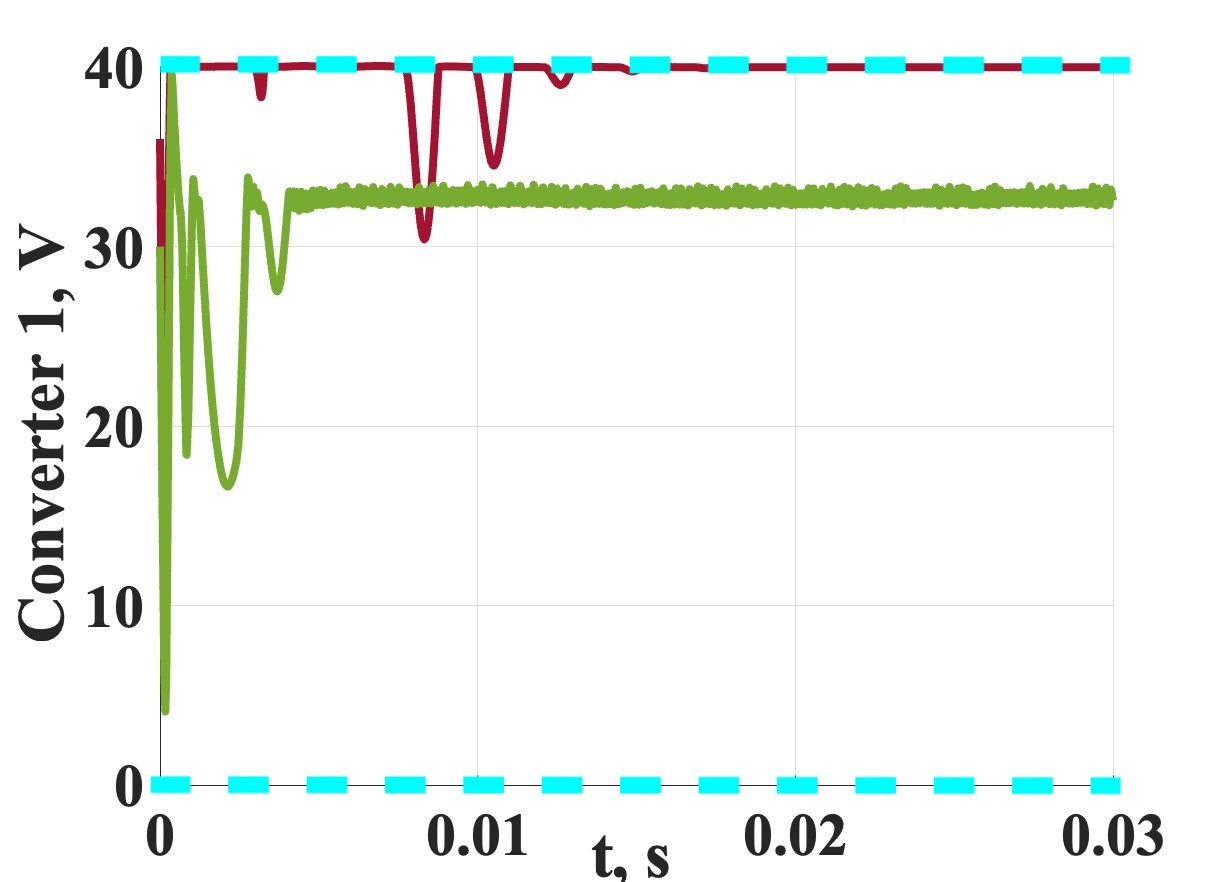}\par 
 \includegraphics[width=\linewidth]{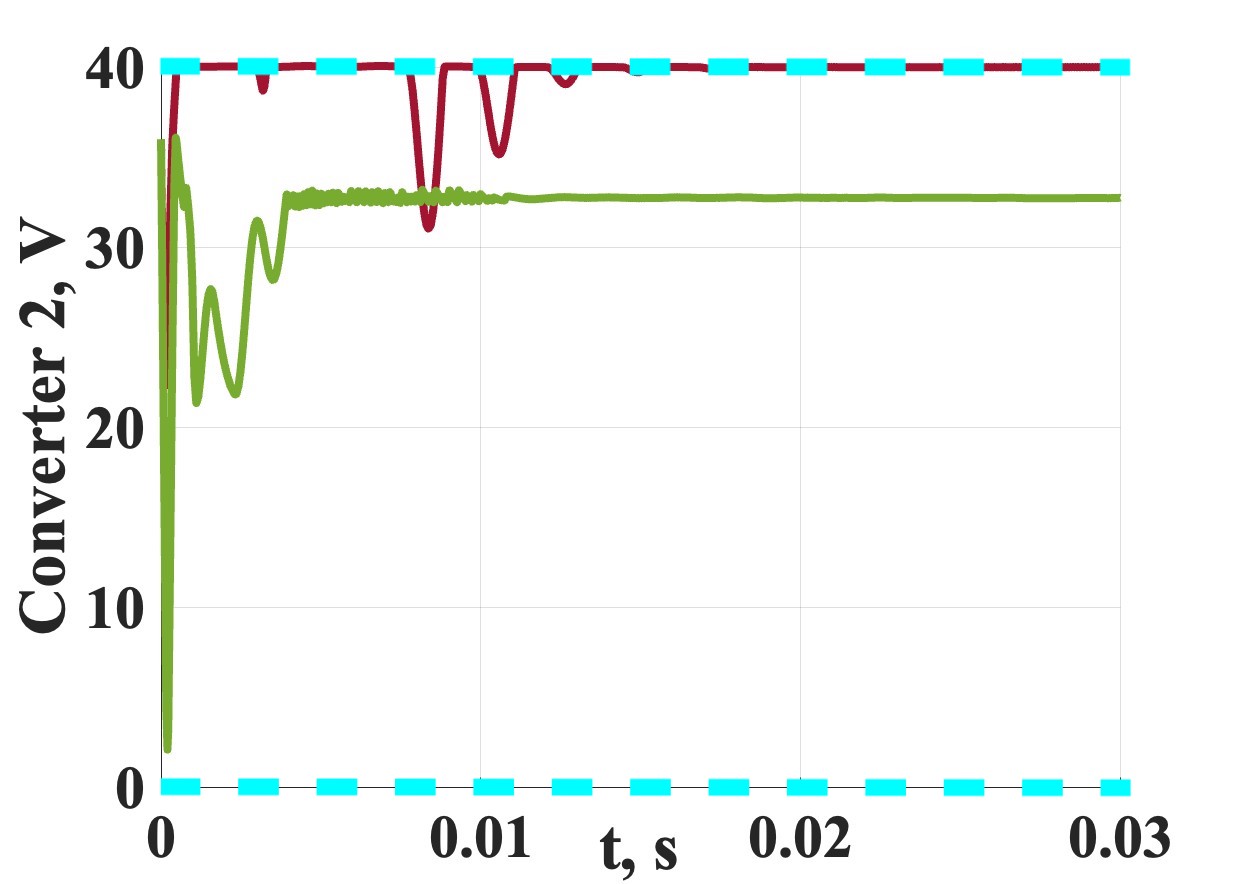}\par    \includegraphics[width=\linewidth]{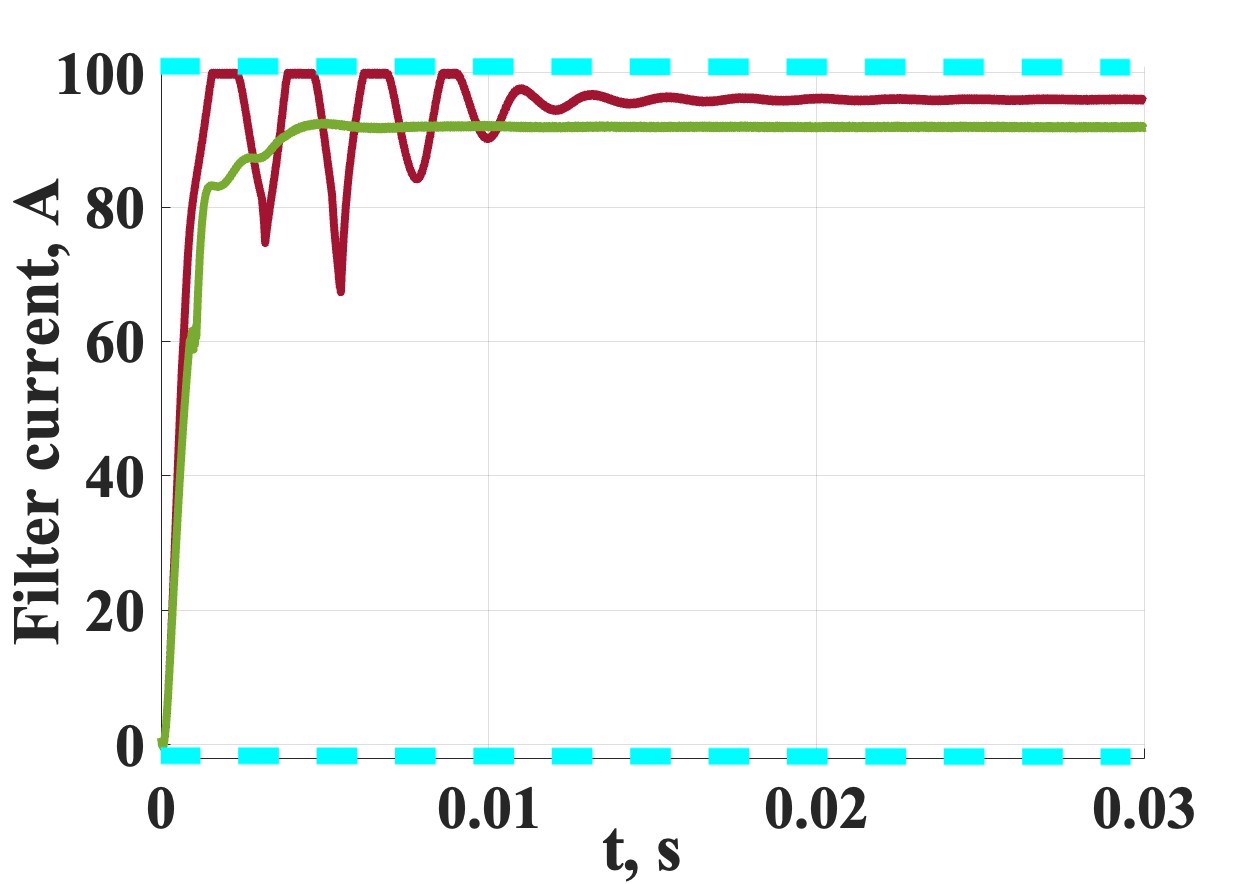}\par 
\end{multicols}
\caption{PLECS validation of the proposed Decentralized Safety-Critical Controller (DSCC) for a DC microgrid in Fig. \ref{fig:plecs_schem} under a cyber attack that shifts the states away from the operating point. To replicate a real scenario, the initial conditions are chosen to be far from the optimal operating point, see Table \ref{tab:my_label1}. The plots compare he performance of the DSCC with the Dynamic IDA-PBC. The objective is to safely regulate the bus of the DC microgrid to its desired steady-state value. Note that only select states and control inputs are shown for conciseness.}\label{fig:plecs_validation}
\end{figure*}
\setlength{\tabcolsep}{3pt}
\begin{table}[ht]
    \centering
    \caption{Results of the numerical experiment tabulated. The desired bus voltage is specified to be 32\,V, and the desired duty cycle of the load-interfacing converter is 50\%. Initial states are picked to be far enough away from the corresponding equilibrium. The steady-state values are averaged over 10$\mu$s in case there are remaining transient harmonics.}\small\begin{tabular}{c|c|c|c|c|c}
    
      State  & Equilibrium   &  Initial  state  &   Steady-state &   Steady-state & Units 
   \\ 
       &   &   & by DSCC & by Dynamic   
       \\
    & & & & IDA-PBC\\
     \hline
    $v_1$ &    32.56 & 23.00& 32.74& 40.00&[V] \\ 
     
    $i_{t_1}$ &  30.01 & 15.00&30.01&29.70&[A] \\\hline
   $v_2 $&    32.56 & 30.00&33.03& 40.00&[V]\\

    $i_{t_2}$ &   31.70 & 12.00 & 31.71& 31.3 &[A]\\\hline
    $v_{b}$ &   32.00 & 1.00&32.32&39.55&[V]   \\
    $i_f$ &   91.43 & 1.00&91.77&49.21&[A]   
     \\
    $v_{l}$ &   16.00 & 9.00&16.09&16.80&[V]   
     \\
    \end{tabular}
    \label{tab:my_label}
\end{table}

\section{Conclusion}
Power electronics-enabled DC microgrids offer an opportunity to move toward more modular and sustainable energy networks in the future. The power efficiency of DC microgrids is well established; however, they lack rigorous cybersecurity and safety certificates. The nonlinear dynamics introduced by DC/DC converters make it difficult to provide the said certificates. In particular, cyberattacks on the microgrid, such as DoS and FDI, can jeopardize both safety and stability. To prevent that, this work proposes a decentralized and online controller named DSCC that can be implemented in practice and comes with provable and reliable guarantees. Each DC/DC converter is controlled locally and does not require any communication with neighboring converters. This modularity ensures that the microgrid is resilient against converter failures caused by DoS attacks. Equally important, DSCC is designed to handle FDI attacks real-time. Performance advantages are verified by high-fidelity switched-circuit simulations. The future direction is to extend the control scheme to different topologies of the DC microgrid.

\ifCLASSOPTIONcaptionsoff
  \newpage
\fi

\bibliographystyle{IEEEtran}
\bibliography{TSG_Murat}

\end{document}